\documentclass[12pt, oneside]{amsart}
\usepackage{latex_base}
\usepackage{macros}
\newcommand{\revise}[1]{{{\color{black}#1}}}
\newif\ifcomment
\commentfalse
\usepackage{array, makecell}

\newcommand{\vartableOne}{
\begin{tabular}{r|rr}
	\multicolumn{3}{c}{$d = 8$, $t = 20$}\\
	$n$ & bit size & time\\
	\hline
	2 & 27723 & 5.82\\
	3 & 23572 & 4.99\\
	4 & 22965 & 4.83\\
	8 & 5678 & 1.12\\
	10 & 1749 & 0.35\\
\end{tabular}}
\newcommand{\vartableTwo}{
\begin{tabular}{r|rr}
	\multicolumn{3}{c}{$d = 10$, $t = 24$}\\
	$n$ & bit size & time\\
	\hline
	2 & 38872 & 8.26\\
	3 & 33938 & 7.26\\
	4 & 31278 & 6.34\\
	8 & 7042 & 1.42\\
	10 & 3778 & 0.91\\
\end{tabular}}
\newcommand{\vartableThree}{
\begin{tabular}{r|rr}
	\multicolumn{3}{c}{$d = 10$, $t = 30$}\\
	$n$ & bit size & time\\
	\hline
	2 & 61198 & 12.41\\
	3 & 57833 & 11.81\\
	4 & 53596 & 10.82\\
	8 & 13343 & 2.55\\
	10 & 6974 & 1.41\\
\end{tabular}}
\newcommand{\tableDegreeTerm}{
\begin{tabular}{c|ccccccc}
	$t\setminus d$ & 6 & 8 & 10 & 18 & 20 & 26 & 28\\
	\hline
	6 & \makecell{912\\0.24} & \makecell{1000\\0.26} & \makecell{1002\\0.26} & \makecell{1170\\0.28} & \makecell{1014\\0.26} & \makecell{955\\0.28} & \makecell{900\\0.26}\\
	\hline
	9 & \makecell{2731\\0.66} & \makecell{2673\\0.65} & \makecell{2808\\0.70} & \makecell{2890\\0.68} & \makecell{2621\\0.61} & \makecell{3166\\0.82} & \makecell{2471\\0.62}\\
	\hline
	12 & \makecell{5599\\1.30} & \makecell{6054\\1.40} & \makecell{5449\\1.27} & \makecell{5747\\1.27} & \makecell{5478\\1.21} & \makecell{6007\\1.53} & \makecell{5027\\1.18}\\
	\hline
	20 & \makecell{9990\\2.26} & \makecell{24078\\5.08} & \makecell{20985\\4.40} & \makecell{21364\\4.43} & \makecell{19324\\3.96} & \makecell{24096\\5.23} & \makecell{17210\\3.59}\\
	\hline
	24 & $\times$ & \makecell{36301\\7.62} & \makecell{33414\\6.99} & \makecell{37080\\7.49} & \makecell{29266\\5.87} & \makecell{37618\\7.87} & \makecell{28090\\5.43}\\
	\hline
	30 & $\times$ & \makecell{57744\\11.90} & \makecell{56354\\11.44} & \makecell{61564\\12.57} & \makecell{48622\\9.32} & \makecell{59975\\12.76} & \makecell{55000\\10.80}\\
	\hline
	50 & $\times$ & $\times$ & \makecell{180971\\36.11} & \makecell{174464\\34.64} & \makecell{146218\\27.80} & \makecell{196511\\38.19} & \makecell{183598\\38.36}\\
	\hline
\end{tabular}}

\newcommand{\tableSoncSageAlt}{
\begin{tabular}{r|rr|rr}
	$t$ & bit size SONC & bit size SAGE & time SONC & time SAGE\\
\hline
	6 & 432 & 1005 & 0.06 & 0.26 \\
	9 & 806 & 2696 & 0.19 & 0.66 \\
	12 & 1261 & 5568 & 0.37 & 1.29 \\
	20 & 2592 & 19203 & 0.64 & 4.00 \\
	24 & 3826 & 32543 & 0.97 & 6.66 \\
	30 & 5029 & 53160 & 1.34 & 10.58 \\
	50 & 10622 & 167971 & 3.95 & 32.78 \\
\end{tabular}}

\newcommand{\degrees}{6,8,10,18,20,26,28}
\newcommand{\Myterms}{6,9,12,20,24,30,50}
\newcommand{\variables}{2,3,4,8,10}
\newcommand{\instances}{2020}
\newcommand{\totalTime}{6780.0}

\newcommand{\roundFaster}{8}

\newcommand{\roundRelLow}{21.6}
\newcommand{\roundRelAvg}{88.6}
\newcommand{\roundRelHigh}{96.8}
\newcommand{\roundBetter}{30}
\newcommand{\roundBad}{256}
\newcommand{\roundOkayPercentage}{81.9}
\newcommand{\roundOkayThreshold}{0.001}
\newcommand{\roundBadThreshold}{1}

\begin{document}

\author{Victor Magron}

\address{Victor Magron, CNRS LAAS, 7 avenue du Colonel Roche, F-31031 Toulouse C\'edex 4, France\medskip}

\email{victor.magron@laas.fr }

\author{Henning Seidler} 

\address{Henning Seidler, Technische Universit\"at Berlin, Institut f\"ur Mathematik, Stra{\ss}e des 17.~Juni 136, 10623 Berlin, Germany\medskip}

\email{seidler@math.tu-berlin.de}

\author{Timo de Wolff}

\address{Timo de Wolff, Technische Universit\"at Berlin, Institut f\"ur Mathematik, Stra{\ss}e des 17.~Juni 136, 10623 Berlin, Germany\medskip}

\email{dewolff@math.tu-berlin.de}



\date{\today}
\title[Exact optimization via SONC and SAGE]{Exact Optimization via Sums of Nonnegative Circuits and Sums of AM/GM Exponentials}
\maketitle

\begin{abstract}
	We provide two hybrid numeric-symbolic optimization algorithms, computing exact sums of nonnegative circuits (SONC) and sums of arithmetic-geometric-exponentials (SAGE) decompositions. Moreover, we provide a hybrid numeric-symbolic decision algorithm for polynomials lying in the interior of the SAGE cone.
	Each framework, inspired by previous contributions of Parrilo and Peyrl, is a rounding-projection procedure. 

	For a polynomial lying in the interior of the SAGE cone, we prove that the decision algorithm terminates within a number of arithmetic operations, which is \revise{polynomial in the number of terms of the input, and linear in the distance  to the boundary of the cone}.
	We also provide experimental comparisons regarding the implementation of the two optimization algorithms. 
\end{abstract}

\section{Introduction}
In this paper, we focus on \struc{\textit{certifying}} the output of polynomial optimization problems in a rigorous way. 
Finding the minimal value of a given polynomial in $n$ variables under polynomial constraints is known to be NP-hard in general \cite{Laurent:Survey}. 
The related problem of deciding \struc{\textit{nonnegativity}} of a polynomial under polynomial constraints is co-NP hard; see e.g., \cite{blum}.
\He{theorem number?}
This decision problem can be solved with the Cylindrical Algebraic Decomposition algorithm \cite{Collins75}, which runs in time doubly exponential in $n$ and polynomial in the maximal total degree $d$ of the input functions.
Further improved algorithms \cite{GV88,BPR98,BGHM3}, relying on critical point methods, allow to decide nonnegativity in singly exponential time in $n$. 
More generally, the complement of the problem lies in the existential theory of the reals, which can be solved in polynomial space and single exponential time \cite{Renegar:PSPACE}.
Safe validation of optimization problem results is  mandatory for guaranted evaluation of mathematical functions \cite{Chevillard11}, certified roundoff error bounds \cite{MagronToms17} or  computer assisted proofs \cite{MagronJFR14,pi17}.

In contrast to the algorithms mentioned above, several numerical frameworks have been developed in the last two decades. 
In the unconstrained case, one way to ensure nonnegativity of a given polynomial $f$ is to decompose $f$ as a \textit{\struc{sum of squares} (SOS)} of polynomials \cite{phdParrilo,Lasserre01},
which provides a certificate that $f$ is nonnegative over the reals. An SOS decomposition can be computed by solving a \struc{\textit{semidefinite program} (SDP)} of size $\binom{n+d}{n}$.
In the constrained case, certificates can be provided by the prominent \struc{\emph{moment-SOS hierarchy}}, also called \struc{\textit{Lasserre's hierarchy}} \cite{Lasserre01,Lasserre:Book:CPOPApplications}. 
Each relaxation is solved with a semidefinite programming solver, implemented in finite-precision arithmetic, whose output is an \emph {approximate} certificate. 
A drawback of these methods is that the size $\binom{n+d}{n}$  of the SDP matrices blows up when the degree $d$ and number of variables $n$ increases.

%
For larger values of $n$, a remedy consists of exploiting a potential sparsity/symmetry pattern arising in the input polynomials.
A sparse version of Lasserre's hierarchy has been developed in~\cite{Waki06SparseSOS,Las06SparseSOS} when the objective function can be written as a sum of polynomials, each of them involving a small number of variables. 
The framework from~\cite{Riener2013SymmetricSDP} allows to take into account the symmetries of the  polynomial optimization problem.
One can also rely on the so-called \struc{\emph{bounded degree SOS hierarchy} (BSOS)}~\cite{lasserre2017bounded}. 
In this hierarchy, positive polynomials are represented as a sum of two terms. 
The first term is an SOS polynomial of degree fixed in advance, while the second one belongs to the set of \struc{\emph{Krivine-Stengle}} representations, that is, is a finite combination of positive scalar weights and cross-products of the polynomials defining the set of constraints.
This allows to handle larger instances than with the standard SOS hierarchy. 
The sparse variant of the BSOS hierarchy~\cite{weisser2018sparse} can handle even bigger problems, under the same sparsity pattern assumptions than the ones used for the sparse SOS hierarchy.
If the support, i.e., number of monomial terms, of the polynomials is small in comparison to $\binom{n+d}{n}$, alternative relaxations based on \struc{\textit{geometric programming (GP)}}~\cite{Duffin:Peterson:Zener:Book}, and, more generally, \struc{\textit{relative entropy programming (REP)}}, potentially allow to obtain lower bounds in a more efficient way than SDP relaxations. 
Both GP and REP are (equivalent to) convex optimization problems over the exponential cone.
These alternative relaxations also provide the possibility to obtain answers when the SDP relaxations cannot be implemented because their size is too large for state-of-the art SDP solvers.

A first class of alternative  certificates is given by \struc{\emph{sums of nonnegative circuit (SONC) polynomials}}.  
A \struc{\textit{circuit polynomial}} is a polynomial with support containing only monomial squares, at the exception of at most one term, whose exponent is a strict convex combination of the other exponents. 
In~\cite{Iliman:deWolff:Circuits}, the authors derive a necessary and sufficient condition to prove that a given circuit polynomial is nonnegative.
When the input polynomial has a more general support, a first attempt is given in~\cite{Ghasemi:Marshall:GP,Ghasemi:Marshall:GP:Semialgebraic} to compute lower bounds while relying on GP.
This approach is generalized in~\cite{Dressler:Iliman:deWolff:SONCApproachConstraints} to compute SONC certificates when the set of constraints is defined as a finite conjunction of polynomial inequalities. In \cite{Dressler:Iliman:deWolff:Positivstellensatz} the authors provide a bounded degree hierarchy, which can be computed via relative entropy programming.
In the recent contribution~\cite{seidler2018experimental}, the second and the third author develop an algorithm computing SONC certificates for sparse unconstrained polynomials with arbitrary support, together with a software library~\cite{poem:software}, called \struc{\textsc{POEM} \textit{(Effective Methods in Polynomial Optimization)}}. 
Although this framework yields a very efficient way to obtain a lower bound for a given polynomial, a drawback is that it currently remains unclear whether the number of circuits involved in a SONC relaxation is exponential in the number of terms of this polynomial. 
A second class of alternative certificates is given by \struc{\emph{sums of arithmetic-geometric-mean-exponentials} (SAGE)} polynomials. 
An \struc{\textit{AGE polynomial}} refers to a \struc{\emph{signomial}}, i.e., a weighted sum of exponentials composed with linear functionals of the variables, which is globally nonnegative with at most one negative coefficient.
The framework from~\cite{Chandrasekaran:Shah:SAGE} derives a hierarchy of convex relaxations providing a sequence of increasing lower bounds for the optimal value of signomial programs (SP). 
For an input polynomial belonging to the SAGE cone, one can compute a SAGE decomposition by solving an REP, involving linear and relative entropy functions.
Furthermore, it is shown in~\cite[Theorem~20]{murray2018newton} that the cones of SAGE and SONC polynomials are related through their equivalence in terms of extreme rays. 
Namely, the extreme rays of the SAGE cone are supported on either a single coordinate or a set of coordinates inducing a simplicial circuit (a circuit with $\ell$ elements containing $\ell - 1$ extreme points). 
Hence, both cones contain the same polynomials.

However, these alternative schemes share the same certification issues than the ones based on SDP relaxations. 
GP/REP/SP solvers rely on interior-point algorithms, implemented in finite-precision. 
Thus, they output only approximate certificates. 
In the unconstrained case, the \struc{\textit{rounding-projection procedure}} by Peryl and Parrilo~\cite{PaPe08} allows to compute a weighted rational SOS decompositon of a polynomial $f$ of degree $d = 2k$, belonging to the interior of the SOS cone. 
In the ``rounding'' step, the algorithm computes an \struc{\emph{approximate Gram matrix}} of $f$, i.e., a matrix $\tilde{\Matrix{G}}$ such that $f \simeq \Vector{v}_k^T \tilde{\Matrix{G}} \Vector{v}_k$, where $\Vector{v}_k$ is the vector of all monomials of degree at most $k$, then rounds $\tilde{\Matrix{G}}$ in the space of rational matrices. 
In the ``projection'' step, the algorithm performs an orthogonal projection to obtain a matrix $\Matrix{G}$, such that $f = \Vector{v}_k^T \Matrix{G} \Vector{v}_k$. With sufficient precision digits, there always exists an SDP matrix fulfilling the above equality, yielding an (exact) Gram matrix associated to $f$.
The last step to retrive an exact weighted SOS decomposition for $f$ consists of performing an exact $\Matrix{L} \Matrix{D} \Matrix{L}^T$ procedure~\cite[\S~4.1]{Golub96}. 
Another framework~\cite{multivsos18,magron2018exact}, provides a hybrid numeric-symbolic framework, which computes exact SOS decompositions under the same assumptions. 
This is based on a ``perturbation-compensation'' algorithm.
In the ``perturbation'' step, one considers an arbitrary small perturbation of the input polynomial, and computes an approximate SOS decomposition with an SDP solver.  
In the ``compensation'' step, one relies on the perturbation terms to recover an exact SOS  decomposition. 
By comparison with the rounding-projection procedure, this algorithm perturbates the input and provides an approximate $\Matrix{L} \Matrix{D} \Matrix{L}^T$ decomposition of the approximate Gram matrix, instead of a projection.
It is shown in~\cite{magron2018exact} that both procedures have a boolean running time, which is singly exponential in $n$ and polynomial in $d$. 
Practical experiments emphasize that the bit size of the SOS outputs obtained with the rounding-projection algorithm is often larger than the one obtained with the perturbation-compensation algorithm. 
The perturbation-compensation algorithm is inspired from  prior work~\cite{univsos}, focusing on weighted SOS decompositions for nonnegative univariate polynomials, where the algorithm from~\cite{Chevillard11} is formalized and analyzed. 
In the unconstrained case, the framework from~\cite{magron2018exact} also extends the perturbation-compensation and rounding-projection algorithms to compute exact Polya and Hilbert-Artin's representations, respectively for positive definite forms and nonnegative polynomials, yielding decompositions into SOS of rational functions, under the assumption that the numerator belongs to the interior of the SOS cone. 
In the constrained case, further algorithms are proposed to derive exact rational Putinar's representations for positive polynomials over basic compact semialgebraic sets.
All corresponding algorithms are integrated in the \struc{$\realcertify$}~\cite{RealCertify} Maple library.

The motivation of the present work is to improve the scalability of these existing certification frameworks, especially for large-size problems, which are currently out of reach when relying on SOS-based methods.

\paragraph{\textbf{Contributions}} 
In this paper, we provide a hybrid numeric-symbolic framework, in a similar spirit as~\cite{PaPe08}, to certify exactly lower bounds obtained after computing SONC/SAGE decompositions with GP/REP relaxations. 
The resulting rounding-projection algorithms allow to handle unconstrained polynomial problems with such exact rational  decompositions.
Our first contributions, given in~\cref{section:post_processing}, are two procedures, called \struc{$\optsonc$} and \struc{$\optsage$}, providing exact rational SONC and SAGE decompositions, respectively.
These two algorithms allow to certify exactly lower bounds of unconstrained polynomials.
Our framework is inspired from~\cite{PaPe08}, first by rounding the output of a given GP/REP relaxation into rational numbers, next by performing an appropriate scaling of these numbers to obtain a solution satisfying exactly the (in)-equality constraints of the relaxation. 
We present another rounding-projection procedure called \struc{$\intsage$} in~\cref{section:exact_SAGE} to handle the case of polynomials belonging to the interior of the SAGE cone.
When the input is an $n$-variate polynomial of degree $d$ with $t$ monomial terms, integer coefficients of maximum bit size $\tau$, and with a distance to the SAGE cone of bit size upper bounded by $\tau_{\sage}$, we prove that Algorithm~$\intsage$ outputs SAGE decompositions within \revise{$\bigo{( (\tau_{\sage}+\tau + t \log t ) \cdot  t^7 \log t )}$ } arithmetic operations.
This is in contrast with the decision algorithm $\intsos$ from~\cite{magron2018exact}, which certifies nonnegativity of  polynomials lying in the interior of the SOS cone in boolean time \revise{$\tau^2 \cdot d^{d^{\bigo{(n)}}}$}.
The two optimization algorithms $\optsonc$ and $\optsage$ are available within the POEM software library. 
In~\cref{section:benchs}, we provide experimental comparisons of these two algorithms. 

\subsection*{Acknowledgements} Victor Magron benefited from the support of the FMJH Program PGMO (EPICS project) and  EDF, Thales, Orange et Criteo.  Timo de Wolff and Henning Seidler are supported by the DFG grant WO 2206/1-1. The authors would like to specially acknowledge the help of Riley Murray for providing insights about the barrier complexity of relative entropy programming.

\section{Preliminaries}
\label{section:background}
Let $\struc{\Z}$ be the set of integers and let $\struc{\R}$, $\struc{\R_{\geq 0}}$ and $\struc{\R_{> 0}}$ be the set of real, nonnegative real and positive real numbers, respectively.
With $\struc{\Q}$ being the set of rational numbers, one defines similarly $\struc{\Q_{>0}}$, $\struc{\Q_{\geq 0}}$.
The \struc{\textit{bit size}} of $i \in \Z$ is denoted by $\struc{\tau(i)} := \lfloor \log_2 (|i|) \rfloor + 1$ with $\struc{\tau(0)} := 1$.  
Given $i \in \Z$ and $j \in \Z \backslash \{0\}$ with gcd$(i,j) = 1$, we define $\struc{\tau(i/j)} := \max \{\tau(i), \tau(j)\}$. 
For two mappings $g,h : \N^l \to \R$, we use the notation \struc{$g(\Vector{v}) \in \bigo{(h(\Vector{v}))}$} to state the existence of $i \in \N$ such that $g(\Vector{v}) \leq i h(\Vector{v})$, for all $\Vector{v} \in \N^l$.
%
%
Throughout the paper, we use bold letters for vectors (small) and matrices (capital), e.g., $\struc{\Vector{x}}=(x_1,\ldots,x_n) \in \R^n$.
For a given vector $\Vector{x}$, we denote the $j$-th coordinate of $\Vector{x}$ by $\struc{x_j} \in \R$, and $\struc{\Vector{x}_{\setminus j}} \in \R^{n-1}$ as the vector obtained by $\Vector{x}$ after removing $x_j$.
Furthermore, let \struc{$\R[\Vector{x}] = \R[x_1,\ldots,x_n]$} be the ring of real $n$-variate polynomials. We define $\struc{\Q[\Vector{x}]}$ similarly.
We denote the set of all $n$-variate polynomials of degree less than or equal to $2d$ by $\struc{\R[\Vector{x}]_{n,2d}}$. 

For a polynomial $p \in \R[\Vector{x}]$, we denote its \struc{\emph{support}} by $\struc{\support{p}} \subset \N^n$; we just write $A$ if it is unambiguous and use the convention $\struc{t} = \# A$.
Thus, $p$ is of the form $\struc{p(\Vector{x})} = \sum_{\Vector{\alpha} \in A}^{} b_{\Vector{\alpha}}\Vector{x}^{\Vector{\alpha}}$ with $\struc{b_{\Vector{\alpha}}} \in \R\setminus\{0\}$ and $\struc{\Vector{x}^{\Vector{\alpha}}} = x_1^{\alpha_1} \cdots x_n^{\alpha_n}$. 
We say, a polynomial is \struc{\emph{sparse}}, if $t \ll \binom{n+2d}{2d} = \dim\left(\R[\Vector{x}]_{n,2d}\right)$.
The support of $p$ can be expressed as an $n \times t$ matrix, which we denote by $\struc{\Matrix{A}}$, such that the $j$-th column of $\Matrix{A}$ is $\Vector{\alpha(j)}$.
Hence, $p$ is uniquely described by the pair $(\Matrix{A},\Vector{b})$, written $p = \struc{\polynomial{\Matrix{A}}{\Vector{b}}}$.
If $\Vector{0}\in A$, then $p(\Vector{0})$ is called the \struc{\emph{constant term}}.

Let us denote by $\struc{\New(p)} := \conv\left(\{\Vector{\alpha} \in \N^n : b_{\Vector{\alpha}} \neq 0\}\right)$  the \textit{\struc{Newton polytope}} of $p$ and $\struc{\vertices{p}} := \left\{\Vector\alpha \in \support{p} : \Vector{\alpha} \text{ is vertex of } \newton{p}\right\}$ be its vertices.
We define the exponents of monomial squares in the support of $p$ as $\struc{\monoSquares{p}} := \left\{\Vector{\alpha} \in \support{p} : \Vector{\alpha} \in (2\N)^n, b_{\Vector{\alpha}} > 0\right\}$.
Its complement in the support is $\struc{\nonSquares{p}} = \support{p} \setminus \monoSquares{p}$.
We indicate the elements of the support which are in the interior of $\New(p)$ by $\struc{\interior{p}} = \support{p} \setminus \partial \New(p)$.

\subsection{SONC Polynomials}
\label{section:background_sonc}
\if{
Geometric programming was introduced in \cite{Duffin:Peterson:Zener:Book}. 
It is equivalent to a convex optimization problem. Applications include problems in circuit design problems, nonlinear network flow, and optimal control; for an overview see \cite{Boyd:Kim:Vandenberghe:TutorialOnGP, Boyd:Vandenberghe:ConvexOptimization}

\begin{definition}
	A function $p : \R_{>0}^n\to \R$ of the form $\struc{p(\Vector{z})} = p(z_1,\ldots,z_n) = cz_1^{\alpha_1}\cdots z_n^{\alpha_n}$ with $c > 0$ and $\alpha_i \in \R$ is called a \struc{\emph{monomial (function)}}. A sum $\struc{\sum_{i=0}^k c_iz_1^{\alpha_{1}(i)}\cdots z_n^{\alpha_{n}(i)}}$ of monomials with $c_i > 0$ is called a \struc{\emph{posynomial (function)}}.

	A \struc{\emph{geometric program} (GP)} has the following form:
	\begin{mini}
		{\Vector{z}\in\R_{>0}^n}
		{p_0(\Vector{z})}
		{\label{prob:mine}}
		{}
		\myCons{p_i(\Vector{z})}{\leq 1}{i=1,\ldots,N}
		\myCons{q_j(\Vector{z})}{= 1}{j=1,\ldots,M}
	\end{mini}
	where $p_0,\dots,p_m$ are posynomials and $q_1,\dots,q_l$ are monomial functions.
	\label{definition:GP}
\end{definition}

The convex equivalents of geometric programs can be solved with interior point methods. 
In \cite{Nesterov:Nemirovskii}, the authors prove worst-case polynomial time complexity of this method; see also \cite[Page 118]{Boyd:Kim:Vandenberghe:TutorialOnGP}.
}\fi
%
We now introduce the fundamental properties of SONC polynomials, as far as they are required in this article. 
SONC are composed of \textit{circuit polynomials}, which were first introduced in \cite{Iliman:deWolff:Circuits}:

\begin{definition}
	\label{definition:CircuitPolynomial}
	A \struc{\emph{circuit polynomial}} $p = \poly(\Matrix{A},\Vector{b}) \in \R[\Vector{x}]$ is of the form $p(\Vector{x}) = \sum_{j=1}^r b_{\Vector{\alpha}(j)} \Vector{x}^{\Vector{\alpha}(j)} + b_{\Vector{\beta}} \Vector{x}^{\Vector{\beta}}$, 
\if{	
	\begin{eqnarray}
		\struc{p(\Vector{x})} & = & \sum_{j=1}^r b_{\Vector{\alpha}(j)} \Vector{x}^{\Vector{\alpha}(j)} + b_{\Vector{\beta}} \Vector{x}^{\Vector{\beta}}, \label{Equ:CircuitPolynomial}
	\end{eqnarray}
}\fi	
	with $0 \leq \struc{r} < n$, coefficients $\struc{b_{\Vector{\alpha}(j)}} \in \R_{> 0}$, $\struc{b_{\Vector{\beta}}} \in \R$, exponents $\struc{\Vector{\alpha}(j)} \in (2\Z)^n$, $\struc{\Vector{\beta}} \in \Z^n$, such that the following condition holds:
	there exist unique, positive \struc{\emph{barycentric coordinates} $\lambda_j$} relative to the $\Vector{\alpha}(j)$ with $j=1,\ldots,r$ satisfying
	\begin{align}
			\label{equ:BarycentricCoordinates}
			& & \Vector{\beta} \ = \ \sum_{j=1}^r \lambda_j \Vector{\alpha}(j) \ \text{ with } \ \lambda_j \ > \ 0 \ \text{ and } \ \sum_{j=1}^r \lambda_j \ = \ 1.
	\end{align}
	For every circuit polynomial $p$ we define the corresponding \struc{\textit{circuit number}} as 
	\begin{align}
		\struc{\Theta_p} \ = \ \prod_{j = 1}^r \left(\frac{b_{\Vector{\alpha}(j)}}{\lambda_j}\right)^{\lambda_j}. 
	\end{align}
\end{definition}
	By Condition \eqref{equ:BarycentricCoordinates}, $\Matrix{A}(p)$ forms a minimal affine dependent set, which is called a \struc{\textit{circuit}}, see e.g., \cite{Oxley:MatroidTheory}.
	More specifically, Condition \eqref{equ:BarycentricCoordinates} implies that $\New(p)$ is a simplex with even vertices $\Vector{\alpha}(1),\ldots,\Vector{\alpha}(r)$ and that the exponent $\Vector{\beta}$ is in the strict interior of $\New(p)$ if $\dim(\New(p)) \geq 1$. 
	Therefore, we call $p_{\Vector{\beta}} \Vector{x}^{\Vector{\beta}}$ the \struc{\emph{inner term}} of $p$.

Since the circuit number alone determines whether they are nonnegative, circuit polynomials are proper building blocks for nonnegativity certificates

\begin{theorem}[\cite{Iliman:deWolff:Circuits}, Theorem 3.8]
	Let $p$ be a circuit polynomial as in~\cref{definition:CircuitPolynomial}. Then $p$ is nonnegative if and only if:
	\begin{enumerate}
		\item $p$ is a sum of monomial squares, or
		\item the coefficient $b_{\Vector{\beta}}$ of the inner term of $p$ satisfies $|b_{\Vector{\beta}}| \leq \Theta_p$.
	\end{enumerate}
	\label{theorem:CircuitPolynomialNonnegativity}
\end{theorem}
The set of \struc{\emph{sums of nonnegative circuit polynomials}} (SONC) is a convex cone. 
For further details about SONC see \cite{deWolff:Circuits:OWR,Iliman:deWolff:Circuits, Dressler:Iliman:deWolff:Positivstellensatz}.

\if{
\begin{definition}
	\label{definition:SONC}
	We define for every $n,d \in \N$ the set of \struc{\emph{sums of nonnegative circuit polynomials} (SONC)} in $n$ variables of degree $2d$ as
	\begin{align*}
		\struc{C_{n,2d}} \ &= \ \left\{f \in \R[\Vector{x}]_{n,2d} \ :\ f = \sum_{\rm finite} p_i,\right.\\&\qquad \left. p_i \text{ is a nonnegative circuit polynomial} \right\}.
		\tag*{\qedhere}
	\end{align*}
\end{definition}
}\fi

Let us consider $p= \poly(\Matrix{A},\Vector{b})  = \sum_{\Vector{\alpha} \in A} b_{\Vector{\alpha}}\x^{\Vector{\alpha}}$. 
To algorithmically determine a lower bound of $p$ via SONC, we take an approach, similar to the one described in \cite[\S~3]{seidler2018experimental}.
As initial relaxation, each monomial non-square is equipped with a negative sign.
This allows us, to restrict ourselves to the positive orthant; see e.g., \cite[Section 3.1]{Iliman:deWolff:Circuits} for further details.
Next, we compute a \struc{\emph{covering}} $\CoverSet$, which is a sequence of sets $\left(\CoverSet^{\Vector{\beta}}\right)_{\Vector{\beta}\in\nonSquares{p}} \subseteq A$ such that $\nonSquares{p}\subseteq \bigcup_{\Vector{\beta}\in\nonSquares{p}} \CoverSet^{\Vector{\beta}}$ and each $\CoverSet^{\Vector{\beta}}$ is the support of a nonnegative circuit polynomial $p_{\Vector{\beta}}$ with interior point $\Vector{\beta}$.
To obtain a covering, we write each non-square as a minimal convex combination of monomial squares, by solving the following LP for each $\Vector{\beta}\in\nonSquares{p}$, as explained in~\cite[Algorithm~3.3]{seidler2018experimental}.
\begin{align}
	\begin{aligned}
		\label{LP:extremal}
		\sum_{\Vector{\alpha}\in\monoSquares{p}} \lambda_{\Vector{\alpha}}^{\Vector{\beta}} \cdot \Vector{\alpha} &= \Vector{\beta}\\
		\sum_{\Vector{\alpha}\in\monoSquares{p}} \lambda_{\Vector{\alpha}}^{\Vector{\beta}} &= 1\\
		\lambda_{\Vector{\alpha}}^{\Vector{\beta}} &\geq 0 \qquad\text{for all } \Vector{\alpha}\in\monoSquares{p}
	\end{aligned}
	\tag{LP}
\end{align}
If $\left\{\Vector{\alpha}: \lambda_{\Vector{\alpha}} > 0\right\}$ is not minimal, then we can reduce it by applying the following lemma.
The lemma is folklore, for a constructive proof see e.g.~\cite[Lemma~3.1]{seidler2018experimental}.
\begin{lemma}
	\label{lem:cover_reduction}
	For every non-extremal point $\Vector{v}\in \support{p}\setminus\vertices{p}$, we can efficiently compute affinely independent $\Vector{v}_0,...,\Vector{v}_m \in \vertices{p}$ with $m\leq n$ such that and $\Vector{v}\in\conv\left(\{\Vector{v}_0,...,\Vector{v}_m\}\right)$.
\end{lemma}
So for each $\Vector{\beta}$, we obtain a vector $\Vector{\lambda}^{\Vector{\beta}}$ of barycentric coordinates, relative to the simplex $\CoverSet^{\Vector{\beta}} = \left\{\Vector{\alpha} \in \monoSquares{p} : \lambda_{\Vector{\alpha}}^{\Vector{\beta}} > 0\right\}$.
We denote these computations by $\struc{\left(\Matrix{\lambda}, \CoverSet\right) := \coverfun{p}}$.

Then, we solve the following geometric program (GP): 
\begin{align}
		\label{problem:SONC}
		\begin{aligned}
			\GPopt = \min_{\Matrix{X}} & \quad \sum_{\Vector{\beta}\in\nonSquares{p}} X_{\Vector{\beta}, \Vector{0}} \\
			\text{s.t.} & \quad \sum_{\Vector{\beta}\in\nonSquares{p}} X_{\Vector{\beta},\Vector{\alpha}} \leq b_{\Vector{\alpha}} \,, \quad \Vector{\alpha} \in \monoSquares{p}, \Vector{\alpha} \neq \Vector{0} \,, \\
			& \quad \prod_{\Vector{\alpha}\in\CoverSet^{\Vector{\beta}}} \left(\frac{X_{\Vector{\beta},\Vector{\alpha}}}{\lambda^{\Vector{\beta}}_{\Vector{\alpha}}}\right)^{\lambda^{\Vector{\beta}}_{\Vector{\alpha}}} = -b_{\Vector{\beta}} \,, \quad  \Vector{\beta} \in \nonSquares{p} \,,\\
			& \quad X_{\Vector{\beta}, \Vector{\alpha}} \geq 0 \,, \quad  \Vector{\alpha} \in \monoSquares{p}, \Vector{\beta} \in \nonSquares{p}\,.
		\end{aligned}
		\tag{SONC}
	\end{align}
For an overview about GPs, see \cite{Boyd:Kim:Vandenberghe:TutorialOnGP, Boyd:Vandenberghe:ConvexOptimization}.
If $\GPopt$ is attained at $\Matrix{\GPsol}$, then one has $p_{\Vector{\beta}} = \sum_{\Vector{\alpha} \in \CoverSet^{\Vector{\beta}}} \GPsol_{\Vector{\beta}, \Vector{\alpha}} \cdot \x^{\Vector{\alpha}} + b_{\beta} \x^{\Vector{\beta}} \geq 0$ by \cref{theorem:CircuitPolynomialNonnegativity}, and $p + \GPopt - b_{\Vector{0}} = \sum_{\Vector{\beta} \in \nonSquares{p}} p_{\Vector{\beta}} \geq 0$. Hence, $b_{\Vector{0}} - \GPopt$ is a lower bound of $p$ on $\R^n$.

These computations correspond to \cite[\S~3.3.2]{seidler2018experimental}, where each cover contains just a single non-square.
The approach dispayed here simplifies the computations, while keeping the quality of the results.

\subsection{SAGE Polynomials}
\label{section:background_sage}
Let $\struc{e} := \exp \, (1)$.
The \struc{\emph{relative entropy}} function is defined for $\nub, \Vector{c} \in \R_+^t$ by $\struc{D(\nub,\Vector{c})} := \sum_{j=1}^t \nu_j \log \frac{\nu_j}{c_j}$. 
A \struc{\emph{signomial}} $p$ is a weighted sum of exponentials composed with linear functionals of a variable $\x \in \R^n$: given $t \in \N$, $c_1,\dots,c_t \in \Q$ and $\ab(1),\dots,\ab(t) \in \N^n$, we write $p(\x) = \sum_{j=1}^t c_j \exp \, ( \ab(j) \cdot \x )$.
%
\if{
\begin{align*}
	p(\x) = \sum_{j=1}^t c_j \exp \, ( \ab(j) \cdot \x ) \,,
\end{align*}
}\fi
Note that for general signomials, one considers $c_1,\dots,c_t \in \R$ and $\ab(1),\dots,\ab(t) \in \R^n$. However, for certification purposes, we restrict the coefficients to the set of rationals and the exponents to tuples of nonnegative integers.
A globally nonnegative signomial with at most one negative coefficient is called an \struc{\emph{AM/GM exponential}} or \struc{\emph{arithmetic-geometric-mean-exponential (AGE)}}.
Certifying the nonnegativity of an AGE is done by verifying an arithmetic-geometric-mean inequality. 
This is recalled in the following result, stated in~\cite[Lemma~2.2]{Chandrasekaran:Shah:SAGE}.
\begin{lemma}
	\label{lemma:age}
	Let $p(\x) = \sum_{j=1}^t c_j \exp \, (\ab(j) \cdot \x ) + \beta \exp \, (\ab(0) \cdot \x )$, with $c_1,\dots,c_t\in \Q_{>0}$, $\beta \in \Q$ and $\ab(0),\ab(1),\dots,\ab(t) \in \N^n$. 
	Then $p(\x) \geq 0$ for all $\x \in \R^n$ if and only if there exists $\nub \in \R_+^t$ such that $D(\nub,e \Vector{c}) \leq \beta$ and $\sum_{j=1}^t \ab(j) \nu_j = (\oneb \cdot \nub) \, \ab(0)$.
\end{lemma}
Given $\ab(0),\ab(1),\dots,\ab(t) \in \N^n$, the set of AGE signomials is a convex cone, denoted by $\Cage$, and defined as follows:
{\small
\begin{align*}
	\struc{\Cage} := \left\{(\Vector{c},\beta) \in \R_+^t \times \R : \text{There exists } \nub \in \R_+^t \text{ with } D(\nub,e \Vector{c}) \leq \beta \,, \sum_{j=1}^t \ab(j) \nu_j = (\oneb \cdot \nub) \, \ab(0) \right\} \,.
\end{align*}
}
The set of \struc{\emph{sums of AGE (SAGE)}} polynomials is also a convex cone, denoted by \struc{$\Csage$}.
By~\cite[Proposition~2.4]{Chandrasekaran:Shah:SAGE}, one has the following characterisation.
\begin{theorem}
	\label{theorem:SAGE_characterisation}
	A signomial $f=\sum_{i=1}^t b_j \exp \, (\ab(j) \cdot \x)$ lies in $\Csage$ if and only if there is $\Vector{c}^{(1)},\dots, \Vector{c}^{(t)}$, $\nub^{(1)},\dots,\nub^{(t)} \in \R^t$ satisfying the following conditions:
	\begin{align}
		\label{problem:SAGE-feas}
		\begin{aligned}
			\sum_{j=1}^t \Vector{c}^{(j)} = \Vector{b} \,, \quad
			\sum_{i=1}^t \ab(i) \nub_i^{(j)} = \Vector{0} \,, \quad
			-\Vector{1} \cdot \nub_{\setminus j}^{(j)} = \nu_j^{(j)} \,, \\
			\Vector{c}_{\setminus j}^{(j)},	\nub_{\setminus j}^{(j)} \geq \Vector{0} \,, \quad
			D\left(\nub_{\setminus j}^{(j)}, e\Vector{c}_{\setminus j}^{(j)}\right) \leq c_j^{(j)} 
		 \,, \quad j=1,\dots,t \,.
		\end{aligned}
		\tag{SAGE-feas}
	\end{align}
\end{theorem}
One way to obtain lower bounds of a signomial $f$ is to solve the following REP:
\begin{align}
		\label{problem:SAGE}		
			f_{\text{SAGE}} = \sup \{ C \in \R: f - C \in \Csage   \} \,.
		\tag{SAGE}
	\end{align}
The constraints of~\cref{problem:SAGE} correspond to~\cref{problem:SAGE-feas}, after replacing $\Vector{b}$ by the vector of coefficients of $f - \lambda$.
\section{Exact Optimization via SONC/SAGE}
\label{section:post_processing}

In this section, we present two algorithms for converting a numerical solution for SONC and SAGE into a lower bound in exact arithmetic.
For a polynomial $p=\sum_{\Vector{\alpha} \in A} b_{\Vector{\alpha}} \Vector{x}^{\Vector{\alpha}}$, we assume $\Vector{0} \in A$, so there exists a constant term $p(\Vector{0}) \neq 0$.
Furthermore, we require that every non-square monomial lies in the interior of $\New{(p)}$ or on a face of $\New{(p)}$ including the origin. 
We denote numerical solutions $\struc{\tilde{\Matrix{X}}}, \struc{\tilde{\Matrix{c}}}, \struc{\tilde{\Matrix{\nu}}}$ with a tilde, intermediate rational solutions $\struc{\hat{\Matrix{X}}}, \struc{\hat{\Matrix{c}}}$ with a hat, and our final rational solution with regular letters.

\subsection{Symbolic Post-Processing for SONC}
\label{subsection:symbolic_post_SONC}

We focus on certifying exactly lower bounds of a given polynomial via SONC decompositions. 
We rely on the numerical procedure from \cref{section:background_sonc}, which starts to compute a covering of this polynomial.
Under these assumptions, we design an algorithm, called $\optsonc$, to convert numerical lower bounds, corresponding to SONC decompositions obtained via GP, into exact lower bounds.

\begin{algorithm}
	$\optsonc$
	\label{algorithm:SONC_post}
	\begin{algorithmic}[1]
		\Require $p=\sum_{\Vector{\alpha} \in A} b_{\Vector{\alpha}} \Vector{x}^{\Vector{\alpha}} \in \Q[\Vector{x}]$,
			rounding precision $\hat{\delta} \in \Q_{>0}$,
			precision parameter $\tilde{\delta} \in \Q_{>0}$ for the GP solver.
		\Ensure Matrix $\Matrix{X}$ of rational numbers, coefficients of the decomposition,
			certified lower bound $C\in\Q$ of $p$ on $\R^n$.
		\State $(\Matrix{\lambda}, \CoverSet) \gets \coverfun{p}$ \label{line:cover}
		\State $\tilde{\Matrix{X}} \gets \gpfun{p, \tilde{\delta}, \Matrix{\lambda}, \CoverSet}$ \label{line:gp} 
			\Comment{Solve \cref{problem:SONC} with accuracy $\tilde{\delta}$}
		\State $\hat{\Matrix{X}} \gets \roundfun{\tilde{\Matrix{X}}}{\hat{\delta}}$ \label{line:roundgp} \Comment{rounding step}
		\For{$\Vector{\alpha} \in \monoSquares{p}$ and $\Vector{\beta} \in \nonSquares{p}$}
		\State $X_{\Vector{\beta},\Vector{\alpha}} \gets b_{\Vector{\alpha}} \cdot \hat{X}_{\Vector{\beta},\Vector{\alpha}} / \sum_{\Vector{\beta'} \in \nonSquares{p}} \hat{X}_{\Vector{\beta'},\Vector{\alpha}}$ \label{line:projsonc} \Comment{projection step}
		\EndFor
		\For{$\Vector{\beta} \in \nonSquares{p}$}
\State $\operatorname{coeff} = \lambda^{\Vector{\beta}}_{\Vector{0}} \cdot\left(-b_{\Vector{\beta}} \cdot \prod_{\Vector{\alpha} \in \CoverSet^{\Vector{\beta}}} \left(\frac{\lambda_{\Vector{\alpha}}^{\Vector{\beta}}}{X_{\Vector{\beta}, \Vector{\alpha}}}\right)^{\lambda_{\Vector{\alpha}}^{\Vector{\beta}}}\right)^{\frac{1}{\lambda_{\Vector{0}}^{\Vector{\beta}}}}$
			\State $X_{\Vector{\beta},\Vector{0}} \gets \roundfunUp{\operatorname{coeff}}{\hat{\delta}}$
			\label{line:constant}
			\Comment{adjust constant term}
		\EndFor
		\State $C \gets b_{\Vector{0}} - \sum_{\Vector{\beta} \in \nonSquares{p}} X_{\Vector{\beta}, \Vector{0}}$
		\State \Return $\Matrix{X}$, $C$
	\end{algorithmic}
\end{algorithm}

In \cref{line:gp}, the function $\gpfun{}$ calls a GP solver to compute a $\tilde{\delta}$-approximation $\tilde{\Matrix{X}}$ of~\cref{problem:SONC}.
This approximation is then rounded in~\cref{line:roundgp} to a rational point $\hat{\Matrix{X}}$ with a prescribed maximal relative error of $\hat{\delta}$.
The projection step from~\cref{line:projsonc} scales the entries of $\hat{\Matrix{X}}$, yielding $\sum_{\Vector{\beta}} X_{\Vector{\beta}, \Vector{\alpha}} = b_{\Vector{\alpha}}$, for all $\Vector{\beta} \in \monoSquares{p}$, to satisfy the first set of equality constraints of~\cref{problem:SONC}.
In \cref{line:constant}, we round the coefficient \emph{up}, with relative error $\hat{\delta}$, so that we have
\begin{align*}
	X_{\Vector{\beta}, \Vector{0}} \geq \lambda^{\Vector{\beta}}_{\Vector{0}} \cdot\left(-b_{\Vector{\beta}} \cdot \prod_{\Vector{\alpha} \in \CoverSet^{\Vector{\beta}}} \left(\frac{\lambda_{\Vector{\alpha}}^{\Vector{\beta}}}{X_{\Vector{\beta}, \Vector{\alpha}}}\right)^{\lambda_{\Vector{\alpha}}^{\Vector{\beta}}}\right)^{\frac{1}{\lambda_{\Vector{0}}^{\Vector{\beta}}}}.
\end{align*}

As in \cref{section:background_sonc}, each $p_{\Vector{\beta}} := \sum_{\Vector{\alpha} \in \CoverSet^{\Vector{\beta}}} X_{\Vector{\beta}, \Vector{\alpha}} \cdot \x^{\Vector{\alpha}} + b_{\beta} \x^{\Vector{\beta}}$ is a nonnegative circuit polynomial. Hence, $C$ is a lower bound for $p$.
%
Our assumption that every circuit polynomial contains a constant term, is necessary to ensure that $\lambda_{\Vector{0}}^{\Vector{\beta}} \neq 0$ for all $\Vector{\beta} \in \nonSquares{p}$ in our computations above.

\subsection{Symbolic Post-Processing for SAGE}
\label{subsection:symbolic_post_SAGE}

Similarly to \cref{algorithm:SONC_post}, our algorithm $\optsage$ takes a given polynomial as input, obtains a numerical lower bound related to a SAGE decomposition computed via REP, and applies a post-processing to find a certified lower bound .
\begin{algorithm}
	$\optsage$
	\label{algorithm:SAGE_post}
	\begin{algorithmic}[1]
		\Require $g = \sum_{i=1}^t b_i \Vector{x}^{\Vector{\alpha(i)}} \in \Q[\Vector{x}]$,
			rounding precision $\hat{\delta} \in \Q_{>0}$, 
			precision parameter $\tilde{\delta} \in \Q_{>0}$ for the REP solver.
		\Ensure Matrices $\Matrix{c}, \Matrix{\nu}$ of rational numbers, coefficients of the decomposition,
			certified lower bound $C\in\Q$ of $g$ on $\R^n$.
		\State $f \gets g(\exp \x - \exp (-\x))$ \label{line:changevars}
		\State Build the $(n+1) \times t$ matrix $\Matrix{Q}$ with columns $(\ab(1),1), \dots, (\ab(t),1)$
		\State $\tilde{\Matrix{c}}, \tilde{\Matrix{\nu}} \gets \repfun{f}{\tilde{\delta}}$ \label{line:optrep}
			\Comment{Solve \cref{problem:SAGE} with accuracy $\tilde{\delta}$}
		\State $\hat{\Vector{c}} \gets \roundfun{\tilde{\Vector{c}}}{\hat{\delta}} \,, \quad \hat{\nub} \gets \roundfun{\tilde{\nub}}{\hat{\delta}}$ \label{line:roundoptsage} \Comment{rounding step}
		\For{$j\in \{1,\ldots,t\}$}
			\State LP $\gets \left\{ \Matrix{Q}\cdot \Vector{\nu}^{(j)} = \Vector{0}, \Vector{\nu}^{(j)}_{\setminus j} \geq \Vector{0}, \|\Vector{\nu}^{(j)}-\tilde{\Vector{\nu}}^{(j)}\|_{\infty}\leq \hat{\delta}, \nu_1^{(j)}\geq\hat{\delta}\right\}$ \label{line:lp}
			\State $\Vector{\nu^{(j)}} \gets$ some element from LP \label{line:projnu}
				\Comment{projection step}			
			\State $\Vector{c}_{\setminus j}^{(j)} \gets \hat{\Vector{c}}_{\setminus j}^{(j)} \quad , c^{(j)}_j \gets b_j - \oneb \cdot \Vector{c}_{\setminus j}^{(j)}$ \label{line:projc} 		
		\EndFor
		\For{$j\in \{1,\ldots,t\}$}
			\State $\operatorname{power} \gets 1-\log\nu_1^{(j)} -\frac{1}{\nu_1^{(j)}} \left(c_j^{(j)} - \sum_{i>1,i\neq j} \nu_i^{(j)}\log\frac{\nu_i^{(j)}}{e c_i^{(j)}}\right)$
			\State $c_1^{(j)} \gets \roundfunUp{\exp\left(\operatorname{power}\right)}{\hat{\delta}}$ \label{line:constantsage}
			\Comment{adjust constant term}
		\EndFor
		\State $C\gets b_1 - \sum_{j=1}^t c_1^{(j)}$
		\State \Return $\Matrix{c}$, $\Matrix{\nu}$, $C$
	\end{algorithmic}
\end{algorithm}

Given a polynomial $g(\y) = \sum_{j=1}^t b_j \y^{\ab(j)}$, one could apply the change of variables $y_i := \exp x_i$ when $\y \in \R^n_{>0}$. 
Since this transformation is only valid on the nonnegative orthant, one workaround used in $\optsage$ is to define the signomial $f(\x) = g(\exp \x - \exp (-\x))$ from~\cref{line:changevars}, in a such a way that a lower bound of $f$ yields a lower bound of $g$.
The $\rep$ function in~\cref{line:optrep} calls an REP solver to compute a $\tilde{\delta}$-approximation $(\tilde{\nub}, \tilde{\Matrix{c}})$ of~\cref{problem:SAGE}.
This approximation is then rounded to a rational point $(\hat{\nub}, \hat{\Matrix{c}})$ with a prescribed maximal relative error of $\hat{\delta}$. 
The projection steps in~\cref{line:projnu} and~\cref{line:projc} ensure that $(\nub, \Vector{c})$ satisfies exactly the linear equality constraints of ~\cref{problem:SAGE}, i.e.,~$\Matrix{Q} \nub^{(j)} = \Vector{0}$ and $\sum_{j=1}^t \c^{(j)} = b$. 
The first projection step boils down to exactly solve an LP with the constraint that $\nu_1^{(j)} > 0$, for all $j=1,\dots,t$, to ensure that further computation in~\cref{line:constantsage} are well-defined. 
Note that this projection could be done while relying on the pseudo-inverse of $\Matrix{Q}$, but one obtains better practical results via this procedure.
To ensure that the relative entropy inequality constraints of~\cref{problem:SAGE} are satisfied, the last step of $\optsage$ aims at finding  $c_j^{(1)}$ such that $c_j^{(j)} \geq D\left(\nub_{\setminus j}^{(j)}, e\Vector{c}_{\setminus j}^{(j)}\right) = 
		 \sum_{i> 1, i\neq j} \nu_i^{(j)} \log\frac{\nu_i^{(j)}}{e c_i^{(j)}} + \nu_1^{(j)} \log\frac{\nu_1^{(j)}}{e c_1^{(j)}}$. 
Thus, one relies on the $\texttt{round-up}$ procedure in~\cref{line:constantsage} to compute $c_1^{(j)} \geq \exp\left(1-\log\nu_1^{(j)} -\frac{1}{\nu_1^{(j)}} \left(c_j^{(j)} - \sum_{i>1,i\neq j} \nu_i^{(j)}\log\frac{\nu_i^{(j)}}{e c_i^{(j)}}\right)\right)$.
Eventually, one has $\sum_{j=1}^t c_i^{(j)}=b_i$, for all $i > 1$ and $\sum_{j=1}^t c_1^{(j)} = b_1 - C$, which certifies that $f - C \geq 0$ on $\R^n$.
We refer to~\cref{appendix:example} for an example of exact SAGE decomposition obtained with $\optsage$.
\section{Deciding Nonnegativity via SAGE}
\label{section:exact_SAGE}

We denote by \struc{$\intCsage$} the interior of the cone $\Csage$ of SAGE signomials. 
A signomial $f=\sum_{j=1}^t b_j \exp \, (\ab(j) \cdot \x)$ lies in $\intCsage$ if and only there is $\Vector{c}^{(1)},\dots, \Vector{c}^{(t)}$, $\nub^{(1)},\dots,\nub^{(t)} \in \R^t$ such that
\begin{align*}
	\label{problem:intSAGE}		
	\begin{aligned}
		\sum_{j=1}^t \Vector{c}^{(j)} = \Vector{b} \,, \quad 
		\sum_{i=1}^t \ab(i) \nub_i^{(j)} = \Vector{0} \,, \quad -\Vector{1} \cdot \nub_{\setminus j}^{(j)} = \nu_j^{(j)} \,, \\
		\Vector{c}_{\setminus j}^{(j)},	\nub_{\setminus j}^{(j)} > \Vector{0} \,, \quad 
		D\left(\nub_{\setminus j}^{(j)}, e\Vector{c}_{\setminus j}^{(j)}\right) < c_j^{(j)} \,, \quad j=1,\dots,t \,.		
	\end{aligned}
	\tag{INTSAGE-feas}
\end{align*}
\revise{
In the case where $f \in \intCsage$, we denote by $\tau_{\sage}(f)$ an upper bound of the bit size of the distance of $f$ to the boundary of $\Csage$.
}

Without the assumptions from~\cref{section:post_processing}, we state and analyze a decision algorithm to certify nonnegativity of signomials belonging to the interior $\intCsage$ of the SAGE cone. 
The resulting hybrid numeric-symbolic algorithm, called $\intsage$, computes exact rational SAGE  decompositions of such signomials.

\if{
For complexity analysis purpose, we recall the following bound for the roots of univariate polynomials with integer coefficients:
\begin{lemma}{~\cite[Theorem~4.2 (ii)]{Mignotte1992}}
	\label{lemma:mignotte}
	Let $f \in \Z[x]$ of degree $d$, with coefficient bit size bounded from above by $\tau$. 
	If $f(\tilde x) = 0$ and $\tilde x \neq 0$, then $\frac{1}{2^\tau + 1} \leq |\tilde x| \leq 2^\tau + 1$.
\end{lemma}
}\fi
%
%
\begin{lemma}
	\label{lemma:sage_distance}
	Let $f = \sum_{j=1}^t b_j \exp \, ( \ab(j) \cdot \x ) \in \intCsage$ of degree $d$ with $\tau = \tau (f)$. 
	Then, there exists $N \in \N$ such that for $\varepsilon := 2^{-N}$, $f - \varepsilon \sum_{j=1}^t \exp \, ( \ab(j) \cdot \x ) \in \Csage$, with \revise{$N = \tau_{\sage}(f)$}.
\end{lemma}
\begin{proof}
	Since $f \in \intCsage$, there are $\Vector{c}^{(1)},\dots, \Vector{c}^{(t)}$, $\nub^{(1)},\dots,\nub^{(t)}$ $\in \R^t$ such that 
	$\sum_{j=1}^t \Vector{c}^{(j)} = \Vector{b}$, 
	$\sum_{i=1}^t \ab(i) \nub_i^{(j)} = \Vector{0}$, 
	$-\Vector{1} \cdot \nub_{\setminus j}^{(j)} = \nu_j^{(j)}$, 
	$\Vector{c}_{\setminus j}^{(j)},	\nub_{\setminus j}^{(j)} > \Vector{0}$ and $D\left(\nub_{\setminus j}^{(j)}, e\Vector{c}_{\setminus j}^{(j)}\right) < c_j^{(j)}$, for all $j=1,\dots,t$. 
	Therefore, there exists $N \in \N$ such that for $\varepsilon := 2^{-N}$, one has $D\left(\nub_{\setminus j}^{(j)}, e\Vector{c}_{\setminus j}^{(j)}\right) + \varepsilon < c_j^{(j)}$, for all $j=1,\dots,t$. 
	For all $i,j=1,\dots,t$, let us define $\mathring{\Vector{b}}$ by $\mathring{b}_i := b_i - \varepsilon$, 
	as well as $\mathring{\Vector{c}}$ by $\mathring{c}_i^{(j)} := c_i^{(j)}$ for $i \neq j$ and $\mathring{c}_j^{(j)} := c_j^{(j)} - \varepsilon$. 
	Note that $\mathring{\Vector{b}}$ is the coefficient vector of $f - \varepsilon \sum_{j=1}^t \exp \, ( \ab(j) \cdot \x )$.	
	Then $\mathring{\Vector{c}}^{(1)},\dots, \mathring{\Vector{c}}^{(t)}$, $\nub^{(1)},\dots,\nub^{(t)}$ satisfy~\cref{problem:intSAGE} after replacing $\Vector{b}$ by $\mathring{\Vector{b}}$, yielding the first claim.

\revise{The second claim follows from the definition of the interior of the SAGE cone.}
\if{
	For the second claim, we start to perform the change of variable $y_i := \exp x_i $, for all $i=1,\dots,n$, define $g(\y,z) := \sum_{j=1}^t b_j \y^{\ab(j)} - z \sum_{j=1}^t \y^{\ab(j)}$, for all $\y \in \R_{>0}^n$.
	It is enough to select $\varepsilon = 2^{-N}$ such that $\varepsilon \leq \inf_{\y \in \R_{>0}^n}  \frac{\sum_{j=1}^t b_j \y^{\ab(j)}}{\sum_{j=1}^t  \y^{\ab(j)}}$. 
	Let us consider the algebraic set $V$ defined by:
	\[
		V := \left\{ (\y,z) \in \R^{n+1} : g(\y,z)  = \frac{\partial g }{\partial y_1} = \dots = \frac{\partial g}{\partial y_n} = 0 \right\} \,.
	\]	
	Using~\cite[Proposition~A.1]{magron2018exact}, there exists a polynomial in $\Z[z]$ of degree less than $(d+1)^{n+1}$ with coefficients of bit size less than $\tau \cdot (4 d + 6)^{3 n + 3}$ such that its set of real roots contains $V$. 
	By~\cref{lemma:mignotte}, it is enough to take $N \leq \tau \cdot (4 d + 6)^{3 n + 3}$, yielding the desired result.
}\fi
\end{proof}

\subsection{Algorithm $\intsage$}
\label{section:intsage}
We present our algorithm $\intsage$ computing exact rational SAGE decompositions for signomials in $\intCsage$.
\begin{algorithm}
$\intsage$
\label{algorithm:intsage}
	\begin{algorithmic}[1]
		\Require $f = \sum_{j=1}^t b_j \exp \, ( \ab(j) \cdot \x ) \in \intCsage$,  rounding precision $\hat{\delta} \in \Q_{>0}$, precision parameter $\tilde{\delta} \in \Q_{>0}$ for the REP solver.
		\Ensure Matrices $\Vector{c}, \nub$ of rational numbers.
		\State Build the $(n+1) \times t$ matrix $\Matrix{Q}$ with columns $(\ab(1),1), \dots, (\ab(t),1)$
		\State $\Matrix{Q}^+ \gets \pseudoinvfun{\Matrix{Q}}$ \label{line:pseudoinv}
		\State ok $\gets$ false
		\While {not ok} \label{line:deltai}
			\State $(\tilde{\Vector{c}}, \tilde{\nub}) \gets \repfun{f}{\tilde{\delta}}$ \label{line:sage}
			\State $\hat{\Vector{c}} \gets \roundfun{\tilde{\Vector{c}}}{\hat{\delta}} \,, \quad \hat{\nub} \gets \roundfun{\tilde{\nub}}{\hat{\delta}}$ \label{line:round} \Comment{rounding step}
			\For {$j \in \{1,\dots,t\}$} \Comment{projection step}
				\State $\nub^{(j)} \gets (\Matrix{I} - \Matrix{Q}^+ \Matrix{Q}) \, \hat{\nub}^{(j)} $ \label{line:proj} 
				\State $\Vector{c}_{\setminus j}^{(j)} \gets \hat{\Vector{c}}_{\setminus j}^{(j)} \,, \quad c^{(j)}_j \gets b_j - \oneb \cdot \Vector{c}_{\setminus j}^{(j)}$
			\EndFor
			\If {for all $j \in \{1,\dots,t\}$, $\nub^{(j)}_{\setminus j}, \Vector{c}^{(j)}_{\setminus j} \geq \Vector{0}$, $c_j^{(j)} \geq D\left(\nub^{(j)}_{\setminus j},e \Vector{c}_{\setminus j}^{(j)}\right)$,  } 
				ok $\gets$ true \label{line:sageok} \Comment{verification step}
			\Else $\ \tilde{\delta} \gets \tilde{\delta}/2$, $\hat{\delta} \gets \hat{\delta}/2$
			\EndIf
		\EndWhile \label{line:deltaf}
		\State \Return $\Vector{c}$, $\nub$
	\end{algorithmic}
\end{algorithm}

The routine $\pseudoinv$ in~\cref{line:pseudoinv} computes the \struc{\textit{pseudo-inverse}} of $\Matrix{Q}$, i.e., a matrix $\Matrix{Q}^+$ such that $\Matrix{Q} \Matrix{Q}^+ \Matrix{Q} = \Matrix{Q}$. 
Next, we enter in the loop starting from~\cref{line:deltai}.  
The $\rep$ function calls an REP solver to compute a $\tilde{\delta}$-approximation $(\tilde{\nub}, \tilde{\Matrix{c}})$ of~\cref{problem:intSAGE}. 
The projection steps ensure that $(\nub, \Vector{c})$ satisfies exactly the linear equality constraints of ~\cref{problem:SAGE-feas}, i.e.,~$\Matrix{Q} \nub^{(j)} = \Matrix{Q}  (\Matrix{I} - \Matrix{Q}^+ \Matrix{Q}) \nub^{(j)} = \Matrix{Q} - \Matrix{Q} \Matrix{Q}^+ \Matrix{Q} = \Vector{0}$ and $\sum_{j=1}^t \c^{(j)} = b$. 
If the inequality constraints are not verified in~\cref{line:sageok}, the rounding-projection procedure is performed again with more accuracy.

\subsection{Arithmetic Complexity}
Before analyzing the arithmetic complexity of~$\intsage$, we first establish lower bounds for the nonnegative components of the solutions related to SAGE decompositions of polynomials in $\intCsage$. 
\begin{lemma}
	\label{lemma:nuc}
	Let $f = \sum_{j=1}^t b_j \exp \, ( \ab(j) \cdot \x ) \in \intCsage$ of degree $d$ with $\tau = \tau (f)$. Let $\varepsilon$ be as in~\cref{lemma:sage_distance}.
	\begin{enumerate}
		\item There exists a solution of  $(\nub,\Vector{c})$ of~\cref{problem:intSAGE} and  $\delta \in \Q_{>0}$ such that $\delta \leq 1$, $(\nub,{\Vector{c}})$ satisfies, $D\left({\nub}_{\setminus j}^{(j)}, e \left({\Vector{c}}_{\setminus j}^{(j)} + \delta \oneb\right) \right) + \frac{\varepsilon}{2} \leq {c}_j^{(j)}$, $\sum_{j=1}^t \Vector{c}^{(j)} = \Vector{b}$, for all $i,j=1,\dots,t$.
		\item There exists a solution $(\nub,\Vector{c})$ of~\cref{problem:intSAGE} and  $\delta \in \Q_{>0}$ such that $(\nub,{\Vector{c}})$ satisfies $D\left({\nub}_{\setminus j}^{(j)} + \delta \oneb , e {\Vector{c}}_{\setminus j}^{(j)} \right) + \frac{\varepsilon}{2} \leq {c}_j^{(j)}$, for all $j=1,\dots,t$.	
		\item There exists a solution $(\nub,\Vector{c})$ of~\cref{problem:intSAGE} and  $\delta \in \Q_{>0}$ such that $(\nub,{\Vector{c}})$ satisfies $D\left((1+ \delta){\nub}_{\setminus j}^{(j)}, e {\Vector{c}}_{\setminus j}^{(j)} \right) + \frac{\varepsilon}{2} \leq {c}_j^{(j)}$, for all $j=1,\dots,t$.
	\end{enumerate}	
\revise{
In each case, $\tau (\delta) \in \bigo{(\tau_{\sage}(f)+\tau + \log t )} = \bigo{(\tau_{\sage}(f)+\tau + n \log d )}$.
}
\end{lemma}

A proof for this lemma is provided in \cref{appendix:proof}.

\begin{theorem}
	\label{th:main} 
	Let $f = \sum_{j=1}^t b_j \exp \, ( \ab(j) \cdot \x ) \in \intCsage$ of degree $d$ and $\tau = \tau (f)$. 
\revise{
There exist $\hat{\delta}$ and $\tilde{\delta}$ of bit size less than $\bigo{(\tau_{\sage}(f)+\tau + t \log t )}$, such that~$\intsagefun{f}{\hat{\delta}}{\tilde{\delta}}$ terminates and outputs a rational SAGE decomposition of $f$ within \[
\bigo{( (\tau_{\sage}(f)+\tau + t \log t ) \cdot  t^7 \log t )}\]
arithmetic operations.
}
\end{theorem}
\begin{proof}
	We first show that the loop of Algorithm~$\intsage$ terminates with $\hat{\delta}$ and $\tilde{\delta}$ of bit size bounded by \revise{$\bigo{(\tau_{\sage}(f)+\tau + t \log t )}$}. 
	Let $\varepsilon$ be as in~\cref{lemma:sage_distance}. When running the procedure~$\rep$, one solves~\cref{problem:SAGE-feas} at precision $\tilde{\delta}$, thus one finds an approximate solution $(\tilde{\nub}, \tilde{\Vector{c}})$ such that $\Vert \sum_{j=1}^t \tilde{\Vector{c}}^{(j)} - \Vector{b} \Vert_\infty \leq \tilde{\delta}$, $ D\left({\tilde{\nub}}_{\setminus j}^{(j)}, e {\tilde{\Vector{c}}}_{\setminus j}^{(j)} \right)   + \varepsilon  \leq \tilde{c}_j^{(j)} + \tilde{\delta}$, and $\Vert \Matrix{Q} \tilde{\nub}^{(j)} \Vert_\infty \leq \tilde{\delta}$, for all $j=1,\dots,t$. 
	After the rounding and projection steps, one obtains $\nub^{(j)} =   (\Matrix{I} - \Matrix{Q}^+ \Matrix{Q}) \hat{\nub}^{(j)}$ and $\Vert \hat{\nub}^{(j)} - \tilde{\nub}^{(j)}\Vert_\infty \leq \hat{\delta}$.
	Since $\ab_j(i) \leq d$, for all $i,j=1,\dots,t$, the bit size of the entries of the matrix $\Matrix{Q}$ is upper bounded by $\tau(d)$.
	Thus, the pseudo-inverse $\Matrix{Q}^+$ has rational entries of bit size bounded by $\bigo{(t \log t + t \log d )} = \bigo{(t \log t)}$, since the bit size is the same as for the determinant length, see~\cite[Corollary~8.13]{BPR98}.
	This implies that the bit size of the difference between the entries of $\tilde{\nub}$ and $\nub$ is upper bounded by $\bigo{\Big( t \log t + \tau(\tilde{\delta}) + \tau(\hat{\delta}) \Big)}$. 
	Similarly, the bit size of the difference between the entries of $\tilde{\Matrix{c}}$ and $\Matrix{c}$ is upper bounded by $\bigo{\Big(\tau +   \tau(\hat{\delta}) \Big)}$. 
	By~\cref{lemma:nuc}, one can perform any absolute or relative perturbation of $\tilde{\nub}$ and $\tilde{\Vector{c}}$, and still ensure that the resulting $(\nub, \Vector{c})$ satisfies $D\left(\nub^{(j)}_{\setminus j},e \Vector{c}_{\setminus j}^{(j)}\right) \leq c_j^{(j)}$, if the perturbation is small enough with bit size at most \revise{$\bigo{(\tau_{\sage}(f)+\tau + \log t )}$}. This implies that one must choose $\tilde{\delta}$ and $\hat{\delta}$ small enough, with an upper bound of \revise{$\bigo{(\tau_{\sage}(f)+\tau + t \log t )}$} on their bit sizes.
	The same reasoning applies to ensure that $\Vector{\nu}_{\setminus j}^{(j)}, \Vector{c}_{\setminus j}^{(j)} \geq \Vector{0}$.
%
	
	Now, we give an upper bound on the number of arithmetic operations. For convex optimization problems having barrier complexity equal to $N$, the standard interior-point methods compute a $\tilde{\delta}$-accurate solution in $\bigo{( \tau(\tilde{\delta}) \sqrt{N}  \log(N) )}$ iterations; see e.g.~\cite[Section~2.4]{Renegar01}. 
	For $(x,y,z) \in \R^3$, the standard barrier complexity of a single relative entropy constraint ``$x \log(x / y) \leq z, x, y \geq 0$'' is equal to 4.  
    In addition, the barrier complexity of a set of constraints is upper bounded by the sum of the complexities of the individual constraints.
    Therefore, the relative entropy formulation given in~\cref{problem:SAGE-feas} has a barrier complexity of $N \leq 4 t^2$.
	At each iteration of the interior-point method, one needs to solve an LP involving $2 t^2$ variables, which can be done within $\bigo{(t^6)}$ arithmetic operations. This yields the upper bound of \revise{$\bigo{( (\tau_{\sage}(f)+\tau + t \log t ) \cdot  t^7 \log t )}$} on the total number of arithmetic operations required while calling~$\intsage$. All other arithmetic operations performed by the algorithm have a negligible cost with respect to the $\rep$ procedure.
%
\end{proof}

\section{Experimental Comparisons}
\label{section:benchs}

We discuss the actual bit sizes and physical running time of $\optsonc$ and $\optsage$ procedures, given by \cref{algorithm:SONC_post} and \cref{algorithm:SAGE_post}.
We describe the setup of our experiment and on which instances the algorithms were tested.
Afterwards, we discuss our findings from running the algorithms on a large set of examples.

\subsection{Experimental Setup}
\label{subsection:experimental:setup}
We give an overview about the experimental setup.

\textbf{Software} 
	The entire experiment was steered by our \textsc{Python} 3.7 based software \textsc{POEM} 0.2.0.0(a) (Effective Methods in Polynomial Optimization), \cite{poem:software}, which we develop since July 2017. 
	\textsc{POEM} is open source, under GNU public license, and available at:
	\begin{center}
		\url{https://www3.math.tu-berlin.de/combi/RAAGConOpt/poem.html}
	\end{center}
	For our experiment, \textsc{POEM} calls a range of further software and solvers for computing the certificates.
	For the numerical solutions of SONC and SAGE, we use \textsc{CVXPY} 1.0.12 \cite{cvxpy}, to create the convex optimization problems, together with the solver \textsc{ECOS} 2.0.7 \cite{ecos}.
	The symbolic computations were done in \textsc{SymPy} 1.3 \cite{Joyner:et:al:SymPy}.

\textbf{Investigated Data} 
	We carried out our experiments on \instances{} randomly generated polynomials.
	The possible numbers of variables are $n=\variables$; the degree takes values $d = \degrees$ and the number of terms can be $t=\Myterms$.
	For each combinations we create instances, where the number of negative terms is one of a few fixed ratios of $t$.
	In particular, the size of \cref{problem:SAGE} grows quadratically in $t$.
	We created the database using \textsc{POEM}, and it is available in full at the homepage cited above.
	Our instances are a subset of those from \cite{seidler2018experimental}.
	In that paper, we also describe their creation in more detail.
	The overall running time for all our instances was \totalTime{} seconds.

\textbf{Hardware and System} 
	We used an \verb+Intel Core i7-8550U+ CPU with 1.8 GHz, 4 cores, 8 threads and 16 GB of RAM under Ubuntu 18.04 for our computations.

\textbf{Stopping Criteria}
	For the accuracy of the solver and the precision of the rounding in \textsc{Python} we used a tolerance of $\varepsilon=2^{-23}$.
	The restriction $t\leq 50$ was chosen, since otherwise we already encounter problem in the numerical solution of \cref{problem:SAGE}.
	The bound $d<30$ was chosen, because for large degree we had a significant increase in the memory required to perform the rounding.
	Both thresholds were obtained experimentally.

\subsection{Evaluation of the Experiment}
\label{subsection:evaluation}

In this section we present and evaluate the results of our experiment and highlight our most important findings, when investigating the computational data.
We focus on the results given by the procedure $\optsage$ (\cref{algorithm:SAGE_post}) via SAGE decompositions and in the end give a comparison to $\optsonc$.

\textbf{Running time \emph{decreases} with growing number of variables.}
	The formulation of \cref{problem:SAGE} shows that the size of the problem only depends on the number of terms $t$, but in the SAGE decomposition, the number of summands is the number of monomial non-squares.
	Most significantly, for more variables, our generating algorithm simply results in a smaller number of these terms.
	Additionally, for $n\geq 8$ and $d \leq 10$, most exponents lie on faces of the Newton polytope.
	This leads to a simpler combinatorial structure, which we believe to result in lower bit sizes and thus in faster solving faster the exact LP from~\cref{line:lp} of~$\optsonc$.
	Next, we have more equality constraints in this LP, which could also improve the running time.
	Lastly, the exponential upper bound is just the worst case, which does not seem to actually happen among our examples.
	\begin{table}[t]
		\begin{minipage}{.23\textwidth}
			\centering
			\vartableOne
		\end{minipage}
			\hfill
		\begin{minipage}{.3\textwidth}
			\centering
			\vartableTwo
		\end{minipage}
		\hfill
		\begin{minipage}{.23\textwidth}
			\centering
			\vartableThree
		\end{minipage}
		\caption{Dependency of the average bit size and the average running time of $\optsage$, with the number of variables, for fixed values of degree $d$ and number of terms $t$;
		For $n=8$ we observe a drastic drop both in running time and bit size.}
		\label{table:vartables}
	\end{table}
	For some selected parameters, we exhibit that behavior in \cref{table:vartables}.
%

\textbf{Dependency of bit size and running time of degree and terms}
	\begin{table}[t]
		\centering
		\small{
			\tableDegreeTerm
		}
		\caption{Bit size (upper part) and running time (lower part) of $\optsage$ in dependency of the degree $d$ and the number of terms $t$ for up to 4 variables;
		A ``$\times$'' indicates, that we do not have instances with these parameters in our data set.}
		\label{table:degree_term}
	\end{table}
	To illustrate how bit size and running time of $\optsage$ vary for different degrees and numbers of terms, we restrict ourselves to at most 4 variables.
	Our numbers from the previous point show, that in these cases bit size and time are similar for fixed $(d,t)$, hence we may aggregate those instances.
	The results are shown in \cref{table:degree_term}.
	We can see that running time and bit size roughly have a linear dependency.
	On the one hand, their growth is quadratic in the number of terms, which matches with the growth of the problem size in \cref{problem:SAGE}.
	On the other hand, bit size and running time are basically unaffected by the degree.
	This shows that the bound, given in the worst case analysis, usually is not met.

\textbf{Quality of the rounding-projection}
	Our experiments verify that in the majority of cases the symbolical lower bound does not diverge far from the numerical bound.
	The detailed distribution is shown in \cref{figure:bar_plot_diff}.
	Most notably, in \roundBetter{} instances, the exact lower bound is even \emph{better} than the numerical bound.
	In \roundOkayPercentage{}\% of the instances, the exact bound differs by at most \roundOkayThreshold{} from the numerical value.
	Only in \roundBad{} instances the difference lies above \roundBadThreshold{}.
	Thus, in the clear majority of examples, the lower bound in exact arithmetic does not differ much from the numerical bound.
	Also, among the instances with large difference, it can also be that the numerical solution actually lies far away from an exact solution.
	So it is unclear, whether a large difference is due to bad behavior of the numerical solution, or a large error in the rounding algorithm.
	\begin{figure}[t]
		\centering
		
	\input{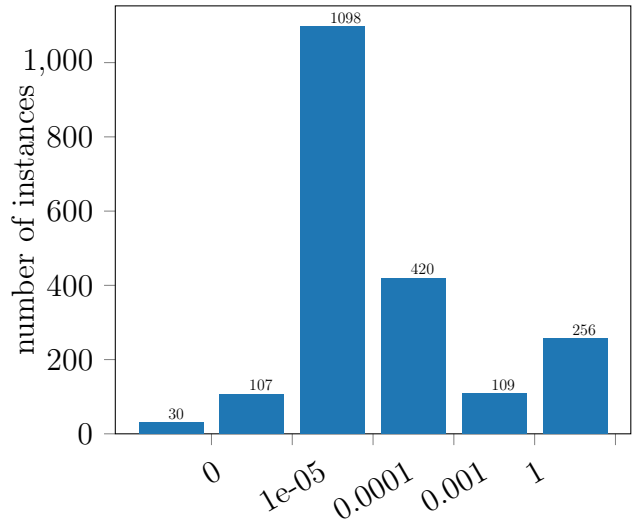}

		\caption{Number of instances where the difference of numerical lower bound and exact lower bound lies in the given interval; note that the exact bound sometimes is better.}
		\label{figure:bar_plot_diff}
	\end{figure}

\textbf{Rounding time versus solving time}
	In nearly every case the rounding procedure takes longer than the numerical solving.
	Only in \roundFaster{} instances, the rounding took less time.
	The ratio of the rounding time to the total time ranges from \roundRelLow{}\% to \roundRelHigh{}\%, with an average of \roundRelAvg{}\%.	
	However, one can implement the rounding procedure much closer to the hardware level, instead of working in Python.
	Thus, we expect that these ratios can be significantly improved.

\textbf{Comparison between SONC and SAGE}
	In their qualitative behavior, $\optsonc$ and $\optsage$ are similar.
	However, $\optsonc$ runs faster and has smaller certificates than $\optsage$, as shown in \cref{table:SONC_SAGE}.
	But one should note that $\optsonc$ only computes \emph{some} lower bound (not necessarily the optimal SONC-bound), whereas $\optsage$ computes the best bound, that can be obtained via this approach.
	Still it shows, that for very large instances, SONC is the method of choice, when other approaches fail due to the problem size.

\textbf{Comparison with SOS} 
	For polynomials lying in the interior of the SOS cone from~\cite[Table~2]{magron2018exact}, we performed preliminary experiments with $\optsage$  and  $\optsage$, which are currently unable to provide nonnegativity certificates. 
	For benchmarks from our database with $n\geq 8$ and $d \geq 10$, $\realcertify$ often fails to provide SOS certificates.
	We plan to provide detailed experimental comparisons with SOS methods in the future.
	
	\begin{table}
		\centering
		\tableSoncSageAlt
		\caption{Comparison of running time and bit size of the certificates between $\optsonc$ and $\optsage$; $\optsonc$ runs faster and has significantly smaller certificates than $\optsage$.}
		\label{table:SONC_SAGE}
	\end{table}

\section{Conclusion and Outlook}
We make two main contributions in this paper.
First, we present an algorithm to decide whether a given multivariate polynomial over the rationals lies in the interior of the SAGE cone.
If that is the case, then the algorithm also computes a certificate in exact arithmetic.
Additionally, we analyze the arithmetic complexity of the algorithm, which is \revise{singly exponential} in the degree and the number of terms, and \revise{doubly exponential} in the number of variables.
Second, we use our numerical methods to obtain lower bounds via SAGE and apply a single iteration of the rounding-projection method, to obtain an exactly certified lower bound.
This method, we run on a large number of test cases.
Based on these experiments, we draw the following conclusions.
\begin{enumerate}
	\item In the majority of cases, the exact solution lies close to the numerical solution, with a difference of at most \roundOkayThreshold{}.
		For few instances, the exact lower bound is even better than the numerical one.
	\item The running time and the bit size grow quadratically in the number of terms, which corresponds to the growth of the problem size.
	\item For the investigated parameters, increasing the degree or the number of variables does not increase the running time or the bit size.
		This also corresponds to the fact, that the size of the REP is independent of the degree and the number of variables.
	\item For very large instances, SONC should be the first choice, to obtain a certified bound, since it runs significantly faster than the other methods.
\end{enumerate}
For future work, the most interesting development would be to have an REP-solver with arbitrary precision, so that we can actually implement $\intsage$ and compare it to similar approaches.
Furthermore, in a significant amount of instances, we encountered computational problem, when calling \optsonc{} or \optsage{}.
So we would like to increase the robustness of our implementation.
Another issue, we have left out so far, is the presence of exponents of monomial non-squares, which lie on a face of the Newton polytope, that does not include the origin.
These result in values $\lambda_{1,j}=0$ (SONC) or $\nu_1^{(j)}=0$ (SAGE), so our computation is undefined.
However, these problems can be circumvented and we plan to do so in a future version of the software. 
Next, we plan to extend our framework to constrained problems and provide more detailed experimental comparisons with SOS-based approaches from~\cite{PaPe08,magron2018exact}, as well as with  methods based on critical points and cylindrical algebraic decomposition. 
Given a polynomial in the interior of the SAGE cone, our decision algorithm $\intsage$ is linear with respect to the distance of this polynomial to the border of the cone.
In order to improve this bound, one could refine the bit size analysis for this  distance. 
A further theoretical aim would be to analyze the boolean running time of $\intsage$, which requires to prove bit complexity estimates for relative entropy optimization problems.

\bibliographystyle{halphainit}
\bibliography{exactsonc}
\newpage
\appendix
\section{Appendix}
\subsection{Proof of~\cref{lemma:nuc}.}
\label{appendix:proof}
	We start with the first claim.
	By~\cref{lemma:sage_distance}, there exist $\tilde{\nub},\tilde{\Vector{c}}$ and $\varepsilon \in \Q_{>0}$, with 
\revise{$\tau(\varepsilon) \in \bigo{ ( \tau_{\sage}(f)  )}$},  such that $(\tilde{\nub},\tilde{\Vector{c}})$ satisfies 	
	$\sum_{j=1}^t \Vector{c}^{(j)} = \Vector{b}$, 
	$\tilde{\Vector{c}}_{\setminus j}^{(j)},	\tilde{\nub}_{\setminus j}^{(j)} > \Vector{0}$, $D\left(\tilde{\nub}_{\setminus j}^{(j)}, e \tilde{\Vector{c}}_{\setminus j}^{(j)}\right) + \varepsilon \leq \tilde{c}_j^{(j)}$, for all $j=1,\dots,t$.
	Let us define $\delta := \frac{\varepsilon}{2 (t-1)}$ and $\Vector{c}$ such that $c_i^{(j)} := \tilde{c}_i^{(j)} + \frac{\varepsilon}{2} $ for all $i\neq j=1,\dots,t$, and $c_j^{(j)} := \tilde{c}_j^{(j)} - t \delta$, for $j=1,\dots,t$. 
	Thus, one has $\sum_{j=1}^t \Vector{c}^{(j)} = \Vector{b}$. For all $j=1,\dots,t$, one has $\nub^{(j)}_{\setminus j} > \Vector{0}$. 
	Combining this together with the fact that the $\log$ function is increasing yields 
	\[
		D\left({\nub}_{\setminus j}^{(j)}, e \left({\Vector{c}}_{\setminus j}^{(j)} + \delta \oneb\right) \right) + \frac{\varepsilon}{2}  \leq D\left({\nub}_{\setminus j}^{(j)}, e {\Vector{c}}_{\setminus j}^{(j)} \right)  \leq  \tilde{c}_j^{(j)} - \frac{\varepsilon}{2}  = c_j^{(j)}  \,.
	\] 
	Using  that for all $d \geq 2$, 
	\begin{align*}
		t &\leq \binom{n+d}{n} = \frac{(n+d) \cdots (d+1)}{n!} =
		\left(1+\frac{d}{n}\right)\left(1+\frac{d}{n-1}\right)\cdots(1+d)\\
		&\leq d^{n-1} (1+d) \leq 2 d^n \,,
	\end{align*}
	one has $\tau(t) \in \bigo{(n \log_2(d))}$.
	Since $\tau(\delta) \leq \tau(\varepsilon) + \tau(t)$, we obtain the first claim.
	
	To prove the second claim, we rely on the following three auxiliary inequalities. 
	
	For $\nu, \delta$ such that $0< \nu < \delta \leq \frac{1}{2}$, one has
	\begin{equation}
		\label{eq:ineq1}
		(\nu + \delta) \log (\nu + \delta) \leq 0 \,,
	\end{equation}
	since $0 < \nu + \delta < 1$ and the function $x \mapsto x \log x $ is negative on $(0, 1)$.
	
	For $\nu, \delta$ such that $0 < \delta \leq \nu$, one has
	\begin{equation}
		\label{eq:ineq2}
		(\nu + \delta) \log (1 + \delta / \nu) \leq 2 \delta \,,
	\end{equation}
	since $\nu + \delta \leq 2 \nu$ and  $\log (1 + \delta / \nu) \leq \delta / \nu$.
	
	For each $\nu, \delta, c > 0$ such that $\nu \geq 1$, one has
	\begin{equation}
		\label{eq:ineq3}
		\delta \log \left(\frac{\nu}{e c}\right) \leq \delta \max \left\{0, \nu \log \left(\frac{\nu}{e c}\right) \right\} \,.
	\end{equation}
	Indeed, if $\nu \leq e c$, the left hand side is less than 0. Otherwise, $\nu \geq 1$ implies that $\log (\frac{\nu}{e c}) \leq \nu \log (\frac{\nu}{e c})$.

	Now, by the first claim, there exist $\varepsilon_{\Vector{c}}, \varepsilon \in \Q_{>0}$, with \revise{$\tau(\varepsilon_{\Vector{c}}), \tau(\varepsilon) \in \bigo{(\tau_{\sage}(f) + \log t )}$}, and $\tilde{\nub},\tilde{\Vector{c}}$ satisfying	
	$\tilde{\nub}_{\setminus j}^{(j)} > \Vector{0}$, $\tilde{\Vector{c}}_{\setminus j}^{(j)} > \varepsilon_{\Vector{c}} \oneb$,   and $D\left(\tilde{\nub}_{\setminus j}^{(j)}, e \tilde{\Vector{c}}_{\setminus j}^{(j)}\right) + \varepsilon \leq \tilde{c}_j^{(j)}$, for all $j=1,\dots,t$. 
	For each $\delta \in \Q_{>0}$ with $\delta \leq \frac{1}{2}$, and all $j=1,\dots,t$, one has:
	\begin{align*}
		D\left(\tilde{\nub}_{\setminus j}^{(j)} + \delta \oneb, e \tilde{\Vector{c}}_{\setminus j}^{(j)}\right) = 
		\sum_{i \neq j} (\tilde{\nu}_i^{(j)} + \delta) \log \frac{\tilde{\nu}_i^{(j)} + \delta}{e \tilde{c}_i^{(j)}}  \,.
	\end{align*}
	We give an upper bound of each summand of the right hande side, depending on the value of $\tilde{\nu}_i^{j}$, for all $i,j=1,\dots,t$ and $i \neq j$. Note that $- \log (e c_i^{j}) \leq \log (e \varepsilon_{\Vector{c}})$, for all $i,j=1,\dots,t$ and $i \neq j$.
	
	If $0< \tilde{\nu}_i^{j} \leq \delta \leq \frac{1}{2}$, one has by~\cref{eq:ineq1}
	\begin{align*}
		(\tilde{\nu}_i^{(j)} + \delta) \log \frac{\tilde{\nu}_i^{(j)} + \delta}{e \tilde{c}_i^{(j)}} 
		&\leq (\tilde{\nu}_i^{(j)} + \delta) \log (\tilde{\nu}_i^{(j)} + \delta) - (\tilde{\nu}_i^{(j)} + \delta) \log (e \tilde{c}_i^{(j)}) \\
		&\leq 2 \delta \log (e \varepsilon_{\Vector{c}})	\,.
	\end{align*}
	If $\delta \leq \tilde{\nu}_i^{j} \leq 1$, one has 
	\begin{align*}
		(\tilde{\nu}_i^{(j)} + \delta) &\log \frac{\tilde{\nu}_i^{(j)} + \delta}{e \tilde{c}_i^{(j)}} \\
		& \leq \tilde{\nu}_i^{(j)} \log \frac{\tilde{\nu}_i^{(j)}}{e \tilde{c}_i^{(j)}} +
		 \delta \log \frac{\tilde{\nu}_i^{(j)}}{e \tilde{c}_i^{(j)}} + (\tilde{\nu}_i^{(j)} + \delta) \log \left(1 + \frac{\delta}{\tilde{\nu}_i^{(j)}}\right)	\\
		& \leq \tilde{\nu}_i^{(j)} \log \frac{\tilde{\nu}_i^{(j)}}{e \tilde{c}_i^{(j)}} + 2 \delta  \log (e \varepsilon_{\Vector{c}}) + 2 \delta \,,
	\end{align*}
	where we use the fact that $\delta \log \tilde{\nu}_i^{(j)} \leq 0$ and bound the last term of the right hand side via~\cref{eq:ineq2}.
	
	If $\tilde{\nu}_i^{j} \geq 1$, we write the first inequality as in the former case and obtain
	\[
		(\tilde{\nu}_i^{(j)} + \delta) \log \frac{\tilde{\nu}_i^{(j)} + \delta}{e \tilde{c}_i^{(j)}} \leq 
		\tilde{\nu}_i^{(j)} \log \frac{\tilde{\nu}_i^{(j)}}{e \tilde{c}_i^{(j)}} + \delta \max \left\{0, \tilde{\nu}_i^{(j)} \log \frac{\tilde{\nu}_i^{(j)}}{e \tilde{c}_i^{(j)}} \right\} + 2 \delta \,,
	\]
	where we rely on~\cref{eq:ineq3} to bound the second term  and the fact that $\delta + \tilde{\nu}_i^{(j)} \leq 2 \tilde{\nu}_i^{(j)}$ together with $ \log \left(1 + \frac{\delta}{\tilde{\nu}_i^{(j)}}\right) \leq \frac{\delta}{\tilde{\nu}_i^{(j)}}$ to bound the last term.
	
	In the worst case, we obtain
	\begin{align*}
	D\left(\tilde{\nub}_{\setminus j}^{(j)} + \delta \oneb, e \tilde{\c}_{\setminus j}^{(j)}\right) & \leq (1+2\delta) D\left(\tilde{\nub}_{\setminus j}^{(j)}, e \tilde{\c}_{\setminus j}^{(j)}\right) + 2 \delta (t - 1)  \vert \log (e \varepsilon_{\c}) \vert + 1) \\
	& \leq  \tilde{c}_j^{(j)} + 2 \delta \tilde{c}_j^{(j)} - (1 + 2 \delta)  \varepsilon + 2  \delta  t \vert \log (e \varepsilon_{\c}) \vert  \\
	& \leq \tilde{c}_j^{(j)}  - (1 + 2 \delta)  \varepsilon + 2 \delta (\vert b_j \vert +  t \vert \log (e \varepsilon_{\c}) \vert)
	\,,
	\end{align*}
	using $\tilde{c}_j^{(j)} = b_j - \oneb \cdot \tilde{c}_{\setminus j}^{(j)} \leq b_j$, thus $\tilde{c}_j^{(j)} \leq \max \{0, b_j \} \leq \vert b_j \vert$.
	
	To ensure that $D\left(\tilde{\nub}_{\setminus j}^{(j)} + \delta \oneb, e \tilde{\c}_{\setminus j}^{(j)}\right) + \frac{\varepsilon}{2} \leq \tilde{c}_j^{(j)}$, it is sufficient to have $- (1 + 4 \delta)  \varepsilon + 4  \delta   (\vert b_j \vert + t \vert \log (e \varepsilon_{\c}) \vert) \leq 0$, which is guaranteed by selecting the largest positive rational $\delta$ such that $\delta \leq \frac{\varepsilon}{4 (\vert b_j \vert + t \vert \log (e \varepsilon_{\c}) \vert - \varepsilon)}$. Since $\tau(\vert b_j \vert ) \leq \tau$, $\tau(t) \in \bigo{(d \log_2 n)}$, $\vert \log (e \varepsilon_{\c}) \vert) \leq \tau (\varepsilon_{\c})$, and $\tau (\varepsilon_{\c}), \tau (\varepsilon) \in \bigo{(\tau_{\sage}(f) + \log t )}$, one can select $\delta$ with bit size at most \revise{$\bigo{(\tau_{\sage}(f) + \tau + \log t )}$}.
	
	The proof of the third claim is very similar and we omit it for the sake of conciseness.
%
\newpage
\subsection{An Example of Exact SAGE Decomposition}
\label{appendix:example}
Let
\begin{align*}
	f(\x) = & \ 277 - 1 x_2^2 + 159 x_2^2 x_3^6 + 275 x_2^4 - 112 x_1^1 x_2^1 x_3^2 + 23 x_1^1 x_2^2 x_3^3 + 338 x_1^2 x_3^4 + 166 x_1^2 x_2^1 x_3^1\\
	&- 89 x_1^2 x_2^1 x_3^2 - 19 x_1^2 x_2^2 x_3^1 + 74 x_1^2 x_2^2 x_3^2 + 268 x_1^6 x_3^2. 
\end{align*}
Our optimization algorithm $\optsage$ returns $(\Matrix{\nu}, \Matrix{c})$ corresponding to the following exact rational SAGE decomposition:
%
$f(\x) = \sum_{j=1}^{12} f_j(\x)$, where $f_j$ is the polynomial with coefficient vector $\Vector{c}^{(j)}$ for $j\in \{1,\dots,12\}$, $f_j = 0$ for $j \in \{1,3,4,7,11,12\}$, and 

\tiny{
\begin{align*}
\begin{aligned}
\Vector{\nu}^{(2)} =& \ \left(\frac{1494563}{131072}, -\frac{1494563}{65536}, 0, \frac{1494563}{131072}, 0, 0, 0, 0, 0, 0, 0, 0\right)\\
f_2 =& \ \frac{7050}{5161} - \frac{47048170074075847768}{2063044386600414657} x_2^2 + \frac{51364821929347737990}{176335433100237218699413} x_2^2 x_3^6 + \frac{1284035788022670567885}{13490375489389305583} x_2^4 +\\ & \frac{83062810507300624}{22200197261503152705103} x_1^1 x_2^1 x_3^2 + \frac{16345503627618}{597428155608650885} x_1^1 x_2^2 x_3^3 + 
\frac{151893109261090080}{12048313292258664398599} x_1^2 x_3^4 + \\
&\frac{6887262866990276056}{3068799030848037634611567} x_1^2 x_2^1 x_3^1 + \frac{584984347065145672}{102911615960676281785795} x_1^2 x_2^1 x_3^2 + \frac{35180881401784}{5417152094682478713} x_1^2 x_2^2 x_3^1 +\\
& \frac{910066980468240}{55380305977850685287} x_1^2 x_2^2 x_3^2 + \frac{7052856072901897195}{672038441901303429467107} x_1^6 x_3^2\\
\Vector{\nu}^{(5)} =& \ \biggl(\frac{2763713}{131072}, \frac{1404885}{65536}, \frac{40783}{131072}, \frac{2618913}{131072}, -\frac{1033291}{8192}, \frac{46873}{65536}, \frac{8086635}{131072}, \frac{2825}{16384}, \frac{16731}{65536}, \frac{5689}{32768}, \frac{4177}{16384},  \frac{3431}{65536}\biggr)\\
f_5 =& \ \frac{4895}{6532} + \frac{397583449413593513}{62516496563648929} x_2^2 + \frac{4282065321967986445380570}{176335433100237218699413} x_2^2 x_3^6 + \frac{5340068537046034265265}{107923003915114444664} x_2^4 - \\
&\frac{797551428730610634080}{6323041088437240873} x_1^1 x_2^1 x_3^2 + \frac{624998409184656940}{119485631121730177} x_1^1 x_2^2 x_3^3 + \frac{2505598243831220279021280}{12048313292258664398599} x_1^2 x_3^4 +\\ 
& \frac{317731662232293486548800}{3068799030848037634611567} x_1^2 x_2^1 x_3^1 + \frac{7986264919496237716952}{20582323192135256357159} x_1^2 x_2^1 x_3^2 + \frac{4894573585100619175}{16251456284047436139} x_1^2 x_2^2 x_3^1 +\\
& \frac{62000736112963044488}{55380305977850685287} x_1^2 x_2^2 x_3^2 + \frac{98753640749442855613080}{672038441901303429467107} x_1^6 x_3^2\\
\Vector{\nu}^{(6)} =& \ \biggl(\frac{24235}{65536}, \frac{11469}{16384}, \frac{731111}{65536}, \frac{413881}{65536}, \frac{75499}{65536}, -\frac{2262417}{65536}, \frac{224607}{65536}, \frac{48937}{65536}, \frac{41999}{32768}, \frac{5895}{4096}, \frac{459131}{65536}, \frac{30411}{32768}\biggr) \\
f_6 =& \ \frac{11}{6124} + \frac{175751812330814062}{6189133159801243971} x_2^2 + \frac{20993411202657969962132460}{176335433100237218699413} x_2^2 x_3^6 + \frac{230789345356723546855}{107923003915114444664} x_2^4 + \\
&\frac{3497123471444427592320}{22200197261503152705103} x_1^1 x_2^1 x_3^2 -  \frac{20624391427747375598}{597428155608650885} x_1^1 x_2^2 x_3^3 + \frac{19031996317336046221320}{12048313292258664398599} x_1^2 x_3^4 +\\
& \frac{188140322704886960215808}{3068799030848037634611567} x_1^2 x_2^1 x_3^1 + \frac{27412383525960721370432}{102911615960676281785795} x_1^2 x_2^1 x_3^2 + \frac{616440189731061280}{1805717364894159571} x_1^2 x_2^2 x_3^1 + \\
&\frac{232955118417446202960}{55380305977850685287} x_1^2 x_2^2 x_3^2 + \frac{239384185705575737209071}{672038441901303429467107} x_1^6 x_3^2\\
\Vector{\nu}^{(8)} =& \ \biggl(\frac{5736955}{131072}, \frac{1718597}{65536}, \frac{2677}{65536}, \frac{283517}{16384}, \frac{298687}{65536}, \frac{5897}{32768}, \frac{1205677}{131072}, -\frac{1398845}{8192}, \frac{338517}{65536}, \frac{488327}{32768}, \frac{37365}{8192},  \frac{1469765}{32768}\biggr)\\
f_8 =& \ \frac{16553}{6389} + \frac{80194631513313539404}{6189133159801243971} x_2^2 + \frac{936097366115717531112852}{176335433100237218699413} x_2^2 x_3^6 + \frac{962838001240718755020}{13490375489389305583} x_2^4 +\\
& \frac{168516279078159280107840}{22200197261503152705103} x_1^1 x_2^1 x_3^2 + \frac{261920871973041930}{119485631121730177} x_1^1 x_2^2 x_3^3 + \frac{622190516331822168107520}{12048313292258664398599} x_1^2 x_3^4 -\\
& \frac{174673447091258679729208576}{1022933010282679211537189} x_1^2 x_2^1 x_3^1 + \frac{1345600783114821896667952}{102911615960676281785795} x_1^2 x_2^1 x_3^2 + \frac{233241950370990311440}{5417152094682478713} x_1^2 x_2^2 x_3^1 +\\ 
& \frac{1847380121666164588230}{55380305977850685287} x_1^2 x_2^2 x_3^2 +  \frac{140920372909036228167589650}{672038441901303429467107} x_1^6 x_3^2\\
\end{aligned}
\end{align*}
}
\newpage
\tiny{
\begin{align*}
\begin{aligned}
\Vector{\nu}^{(9)} =& \ \biggl(\frac{942879}{131072}, \frac{66451}{8192}, \frac{13687}{131072}, \frac{291521}{32768}, \frac{487245}{65536}, \frac{29821}{65536}, \frac{2150473}{65536}, \frac{252191}{32768}, -\frac{3460417}{32768}, \frac{170395}{16384}, \frac{548667}{65536}, \frac{925733}{65536}\biggr)\\
f_9 =& \ \frac{181}{1076} + \frac{3266564780864805056}{2063044386600414657} x_2^2 + \frac{945568751423587451936415}{176335433100237218699413} x_2^2 x_3^6 + \frac{195552867021926712375}{13490375489389305583} x_2^4 +\\
& \frac{108598341793757685149280}{22200197261503152705103} x_1^1 x_2^1 x_3^2 + \frac{261624143096996295}{119485631121730177} x_1^1 x_2^2 x_3^3 + \frac{876811143650234639492646}{12048313292258664398599} x_1^2 x_3^4 + \\
&\frac{9330362356036127696945600}{3068799030848037634611567} x_1^2 x_2^1 x_3^1 - \frac{10867831413671203786608456}{102911615960676281785795} x_1^2 x_2^1 x_3^2 + \frac{192909634965734635520}{16251456284047436139} x_1^2 x_2^2 x_3^1 + \\
& \frac{1339551331036614124520}{55380305977850685287} x_1^2 x_2^2 x_3^2 + \frac{17532007073614874704137330}{672038441901303429467107} x_1^6 x_3^2 \\
\Vector{\nu}^{(10)} =& \ \biggl(\frac{732091}{262144}, \frac{167933}{32768}, \frac{14541}{131072}, \frac{7773365}{262144}, \frac{1327}{512}, \frac{29683}{65536}, \frac{272641}{131072}, \frac{293905}{65536}, \frac{214965}{65536}, -\frac{4886959}{65536}, \frac{144047}{32768},  \frac{1284635}{65536}\biggr)\\
f_{10} =& \ \frac{157}{2744} + \frac{5424537902241772912}{6189133159801243971} x_2^2 + \frac{880139855950527034906380}{176335433100237218699413} x_2^2 x_3^6 + \frac{571069367996398273530}{13490375489389305583} x_2^4 +\\ 
&\frac{33169145578648933133280}{22200197261503152705103} x_1^1 x_2^1 x_3^2 + \frac{228162076394260360}{119485631121730177} x_1^1 x_2^2 x_3^3 + \frac{48697840759706172793616}{12048313292258664398599} x_1^2 x_3^4 +\\
& \frac{1587820308255205569373114}{1022933010282679211537189} x_1^2 x_2^1 x_3^1 + \frac{59150503389680767180777}{20582323192135256357159} x_1^2 x_2^1 x_3^2 - \frac{403951932103643744176}{5417152094682478713} x_1^2 x_2^2 x_3^1 + \\
&\frac{616254425060782282800}{55380305977850685287} x_1^2 x_2^2 x_3^2 + \frac{21315777567587124730738350}{672038441901303429467107} x_1^6 x_3^2\\
\end{aligned}
\end{align*}
}

\end{document}